\documentclass[a4paper,11pt]{article}

\usepackage{draft}
\usepackage{amsmath}
\usepackage{amssymb}
\usepackage{amsthm}
\usepackage{mathtools}
\usepackage{mathrsfs}
\usepackage{booktabs}
\usepackage{tikz}
\usetikzlibrary{arrows.meta,calc}
\usepackage{microtype}
\usepackage{orcidlink}
\usepackage{cite}
\usepackage{xcolor}
\definecolor{twilightlavender}{rgb}{0.54, 0.29, 0.42}
\definecolor{richmaroon}{rgb}{0.69, 0.19, 0.38}
\definecolor{forestgreen(web)}{rgb}{0.13, 0.55, 0.13}
\definecolor{lava}{rgb}{0.81, 0.06, 0.13}
\hypersetup{
  breaklinks,
  colorlinks,
  citecolor=forestgreen(web),
  filecolor=richmaroon,
  linkcolor=lava,
  urlcolor=twilightlavender,
  pdftitle={Reflected constraint algebras of open Carroll--Weyl null strings},
  pdfauthor={Limin Zeng},
  pdfsubject={Open Carroll--Weyl null strings and reflected constraint algebras},
  pdfkeywords={Carrollian symmetry, open null strings, boundary conformal
    algebra, constrained Hamiltonian systems}
}

\newcommand{\CW}{\mathrm{CW}}
\newcommand{\dd}{\mathrm{d}}

\newcommand{\dchi}{\delta_\chi}
\newcommand{\leftcell}[1]{\parbox[t]{\linewidth}{\raggedright #1}}
\newtheorem{proposition}{Proposition}

\numberwithin{equation}{section}
\renewcommand{\theequation}{\arabic{section}.\arabic{equation}}

\title{Reflected constraint algebras of open Carroll--Weyl null strings}

\affiliation[a]{
  School of Fundamental Physics and Mathematical Sciences,\\
  Hangzhou Institute for Advanced Study, UCAS, Hangzhou 310024, China
}
\affiliation[b]{
  Institute of Theoretical Physics, Chinese Academy of Sciences,
  Beijing 100190, China
}
\affiliation[c]{
  University of Chinese Academy of Sciences, Beijing 100049, China
}

\author[a,b,c \orcidlink{0009-0003-5703-2687}]{Limin Zeng}
\emailAdd{zenglimin25@mails.ucas.ac.cn}

\abstract{We study whether open null strings admit the full Carroll--Weyl
gauge symmetry consistently with endpoint conditions and nontrivial physical
degrees of freedom.  For endpoints constrained to timelike defects, natural
variations imply that any nonzero Carroll flow must be tangent to the
worldsheet boundary.  For a fixed linear defect through the target origin, we
construct a differentiable Hamiltonian system on this branch.  Its complete
Carroll--Weyl generator requires no additional boundary charge.  Smooth
reflection defines a preserved subdomain and gives a circle interpretation to
selected sine and cosine modes of the constraints.  Their centreless algebra
extends the two-constraint open-string algebra by the local scale constraint.
An endpoint obstruction shows that this reflection condition
excludes some gauge orbits of the natural interval theory: we exhibit a
regular four-dimensional solution with no corresponding reflected
representative.  Within the reflected sector, we construct a local canonical
field pair and conserved target-space Lorentz charges.  A separate solution
family has nonzero boost charge and a nondegenerate two-dimensional symplectic
pullback.  The constraint algebra, the boundary choice, and the physical
observables thus provide complementary information about the open theory.}


\begin{document}
\frenchspacing
\maketitle

\section{Introduction}
\label{sec:introduction}

Tensionless strings provide a direct bridge from string dynamics to
particle-like Carrollian kinematics.  In the zero-tension limit, the
Lorentzian worldsheet is replaced by a Carrollian geometry
\cite{CarrollGeometry}.  Recent closed-string analyses show that gauging the
complete Carroll--Weyl (CW) flow is necessary for a consistent local description
\cite{Overlooked,CWI,CWII,WhatClassical}.  On an interval, bulk covariance
does not determine the phase-space domain, the differentiability of a
generator, or the preservation of boundary data.  These are Hamiltonian
questions.

Classical null strings were introduced by Schild \cite{Schild}.  The
Isberg--Lindstr\"om--Sundborg--Theodoridis (ILST) action realizes worldsheet
diffeomorphisms and common Weyl scaling \cite{ILST}.  Carrollian symmetries
are reviewed in \cite{CarrollReview,RuzziconiReview}.  The local
volume-preserving CW flow requires a connection and a third constraint
\cite{Overlooked,CWI,CWII}.  Earlier Hamiltonian conformal-string treatments
already contained local scale connections and a related constraint structure
\cite{GustafssonConformalString1995,LindstromTensionlessHistory2026}.  The
model below uses the CW system and does not impose the additional
\(X^2=0\) condition of that formulation.

Open null strings were constructed in the ungauged ILST description
\cite{OpenNull}.  Their boundary algebra is the boundary Carrollian conformal
algebra (BCCA) \cite{BCCAMath}.  Boundary Carrollian constructions have also
been extended to supersymmetric settings, yielding homogeneous and
inhomogeneous boundary superconformal algebras and an open null-superstring
realization of the homogeneous case \cite{Bagchi:2025BSCCA}.  Tensionless
transitions and oscillator gluing were studied with a Kalb--Ramond field
\cite{TensionlessOpen}.  The open CW problem is whether the third constraint,
the connection, and the auxiliary-field transformations admit one common
reflected interval domain.

Related studies extend the closed theory to curved and AdS targets
\cite{AdSCW,AdSNull,ConformalNull,CarrollTargets}, discuss holographic formulations
\cite{NullHolo,BRSTDescent,CWFlatHolo}, and analyze classical or quantum
anomalies or light-cone quantization
\cite{CWBRS,CWGhosts,LCQuant,CWAnomalies,QuantumCarrollString}.  Carroll
D\(_p\)-brane worldvolume constructions appear in
\cite{ZengDW,KlusonDBrane}, while null \(p\)-branes have been analyzed in a
light-cone Hamiltonian framework \cite{NullPBranes}.  Boundary-state methods
in conformal field theory
\cite{Cardy1989} and boundary degrees of freedom in gauge theory, gravity, and
string theory
\cite{DonnellyFreidel2016,DonnellyWong2017,BallLawWong2024,MaccaferriBoundaryModes2025}
provide related perspectives on boundary sectors.  None fixes the reflected
endpoint domain or the induced boundary algebra.

Our first result identifies the endpoint geometry required by a nonzero
massless flow on a timelike defect.  Natural tangential variations force the
Carroll direction to run along the worldsheet boundary; we call such an
endpoint characteristic.  For a fixed linear defect through the origin, the
resulting interval admits differentiable Hamiltonian generators and the
complete Carroll--Weyl transformation without an additional CW boundary
charge.

Our second result determines what smooth reflection selects.  It supplies a
preserved domain and a circle description of the constraint modes, while
excluding gauge orbits with a nonzero endpoint spatial gradient of the
tangent momentum direction.  A regular four-dimensional solution realizes
this obstruction explicitly.

Our third result exhibits the surviving physical content.  Solving the
constraints on a fixed-endpoint-angle layer gives one canonical field pair.
A separate boost family varies gauge-invariant endpoint directions and
carries a nonzero Lorentz charge canonically paired with its rapidity.

\paragraph{Setup.}
We work with the classical Carroll--Weyl null-string action in flat target
space, without additional endpoint interactions.  The Hamiltonian
construction fixes a nondegenerate linear defect through the origin and uses
the chart \(V^0\neq0\).  Smooth reflection is introduced in
Sec.~\ref{sec:algebra} as an additional choice of boundary data.

\begin{table}[ht]
\centering
\small
\begin{tabular}{@{}p{0.22\textwidth}p{0.31\textwidth}p{0.37\textwidth}@{}}
\hline
domain & additional condition & question answered\\
\hline
\(\widehat\Gamma_\partial\) &
\(X^I=P_I=V^1=0\) at both endpoints and \(V^0\neq0\) &
differentiability, tangency, and local time preservation\\
\(\mathcal R_R\) &
smooth periodic reflection with all endpoint jets and auxiliary parities &
circle projection and preservation of reflection\\
fixed-angle chart &
regular reflected layer with the three gauge conditions &
one local canonical field pair\\
boost family &
same reflection with varying endpoint momentum direction &
nonzero charge and two-dimensional symplectic pullback\\
\hline
\end{tabular}
\caption{Domains used in the paper.  The constraints and their gauge flows act
on these boundary domains.}
\label{tab:domains}
\end{table}

Sec.~\ref{sec:action} derives the natural endpoint conditions.
Sec.~\ref{sec:canonical} constructs the characteristic interval Hamiltonian
system.  Sec.~\ref{sec:algebra} develops the constraint modes and their
reflected realization.  Sec.~\ref{sec:regular} presents the reduced chart and
physical observables.  Appendices~\ref{app:half-interval}
and~\ref{app:closure} contain the mode and bracket details, while
Appendix~\ref{app:reflection-scope} proves the reflection-orbit obstruction.

\subsection*{Relation to prior constructions}

The open ILST theory has two bulk constraints and the BCCA boundary algebra
\cite{OpenNull,BCCAMath,NullReview}.  The closed CW theory adds a connection
and a third first-class constraint \cite{CWI,CWII}.  The open question is
whether both survive the endpoint.  We adapt the closed bulk brackets to the
interval and show that the \(D-3\) count survives on the regular constraint
surface.  The endpoint flux requires the phase-space domain and admissible
generators to be specified independently of bulk constraint closure.  The
orbit obstruction below shows that the fixed-defect reflected domain does not
cover all regular natural interval data.

The reflected sector is specified by a constant metric-preserving target
involution, compatible auxiliary fields, and reflection-compatible gauge
parameters.  Its selected mode algebra extends the two-constraint open ILST
sector by the CW scale current.  Table~\ref{tab:comparison} compares these
constraint structures and their regular local embedding-pair counts.
Appendix~\ref{app:central-extensions} separately studies candidate central
terms within a specified local diagonal-contact class, with the zero-mode
actions held fixed.

\begin{table*}[t]
\centering
\small
\begin{tabular}{@{}p{0.23\textwidth}p{0.19\textwidth}
p{0.22\textwidth}p{0.25\textwidth}@{}}
\hline
\leftcell{model} & \leftcell{third CW constraint} &
\leftcell{boundary algebra} & \leftcell{regular local embedding pairs}\tabularnewline
\hline
\leftcell{open ILST} & \leftcell{no} & \leftcell{BCCA} &
\leftcell{\(D-2\)}\tabularnewline
\leftcell{closed CW} & \leftcell{yes} & \leftcell{no boundary} &
\leftcell{\(D-3\)}\tabularnewline
\leftcell{this paper} & \leftcell{yes} &
\leftcell{selected constraint mode algebra} & \leftcell{\(D-3\)}\tabularnewline
\hline
\end{tabular}
\caption{Comparison of the known two-constraint open null string, the
closed CW theory, and the present open CW completion.  Here open ILST means
the ungauged two-constraint construction of Ref.~\cite{OpenNull}.  The
selected constraint mode algebra has both natural-interval and reflected
representations.  The local pair counts assume functional independence of
the constraints.}
\label{tab:comparison}
\end{table*}

\section{Endpoint conditions and the characteristic sector}
\label{sec:action}

\subsection{The gauged action and endpoint flux}

The field \(V^a\) specifies the Carroll flow on the worldsheet, and \(W_a\)
is the connection for local Carroll--Weyl scaling.  Our first goal is to
extract the endpoint term in the variation of the action.  This term
separates conditions required by the variational principle from the reflected
data selected later.

Let \(X^\mu\) embed the interval worldsheet into a flat \(D\)-dimensional
Minkowski target with metric \(\eta_{\mu\nu}\) of signature
\((D-1,1)\), and let the dimensionful normalization \(\kappa\) be nonzero.
It is the coefficient of the null-string action, not a restored finite
string tension.  Define
\begin{equation}
  D_aX^\mu=\partial_aX^\mu+W_aX^\mu.
  \label{eq:covariant}
\end{equation}
The transformations are
\begin{equation}
  \dchi X^\mu=\chi X^\mu,
  \qquad
  \dchi V^a=-\chi V^a,
  \qquad
  \dchi W_a=-\partial_a\chi.
  \label{eq:transformations}
\end{equation}
The Carroll field \(V^a\) is a contravariant worldsheet vector density of
weight \(1/2\).  Under a coordinate change it transforms as
\begin{equation}
  V'{}^a(\sigma')
  =
  |J|^{-1/2}
  \frac{\partial\sigma'^a}{\partial\sigma^b}V^b(\sigma),
  \qquad
  J=\det\frac{\partial\sigma'}{\partial\sigma}.
  \label{eq:carroll-density}
\end{equation}
This diffeomorphism-density weight is distinct from the CW weight minus one.
For a fixed-point reflection, \(|J|=1\) and the reflection parity of \(V^1\)
is unchanged.
The derivative of the parameter cancels between the coordinate and
connection transformations.  The resulting covariance identity,
\begin{equation}
  \dchi D_aX^\mu=\chi D_aX^\mu,
  \label{eq:covariant-transformation}
\end{equation}
is the standard closed CW identity \cite{CWI,CWII}.

The gauged action is
\begin{equation}
  S_{\CW}
  =
  \frac{\kappa}{2}
  \int_\Sigma\dd^2\sigma\,
  V^aV^bD_aX^\mu D_bX_\mu.
  \label{eq:action}
\end{equation}
The Carroll-vector factor has weight minus two, while each covariant
derivative has weight plus one.  The action integrand is therefore neutral
and
\begin{equation}
  \dchi S_{\CW}=0
  \label{eq:integrand-invariance}
\end{equation}
pointwise.  CW invariance does not require a cancelling boundary term on the
worldsheet interval.

The endpoint flux is
\begin{equation}
  F^a{}_\mu=\kappa V^aV^bD_bX_\mu.
  \label{eq:flux}
\end{equation}
For an embedding variation at fixed \(V\) and \(W\),
\begin{equation}
  \delta_XS_{\CW}
  =
  \kappa\int_\Sigma\dd^2\sigma\,
  V^aV^bD_a\delta X^\mu D_bX_\mu
  =
  \int_\Sigma\dd^2\sigma\,
  F^a{}_\mu D_a\delta X^\mu .
  \label{eq:x-variation-expanded}
\end{equation}
The density \(F^a{}_\mu\) has weight minus one.  Its covariant divergence is
\begin{equation}
  D^{(-1)}_aF^a
  =
  \partial_aF^a-W_aF^a .
  \label{eq:weighted-divergence}
\end{equation}
Therefore
\begin{equation}
  F^aD_a\delta X
  =
  \partial_a(F^a\delta X)
  -
  D^{(-1)}_aF^a\,\delta X .
  \label{eq:weighted-ibp}
\end{equation}
Integrating this identity gives
\begin{equation}
  \delta_XS_{\CW}
  =
  \delta_XS_{\rm bulk}
  +
  \kappa\int_{\partial\Sigma}\dd\tau\,
  n_aV^aV^bD_bX_\mu\delta X^\mu.
  \label{eq:boundary-variation}
\end{equation}
The normal Carroll contraction is retained.  It vanishes at a characteristic
endpoint.  Here \(n_aV^a\) is the normal component of the Carroll flow across
the worldsheet boundary.

\subsection{Natural endpoint conditions and nonzero-flow necessity}
\label{sec:boundaries}

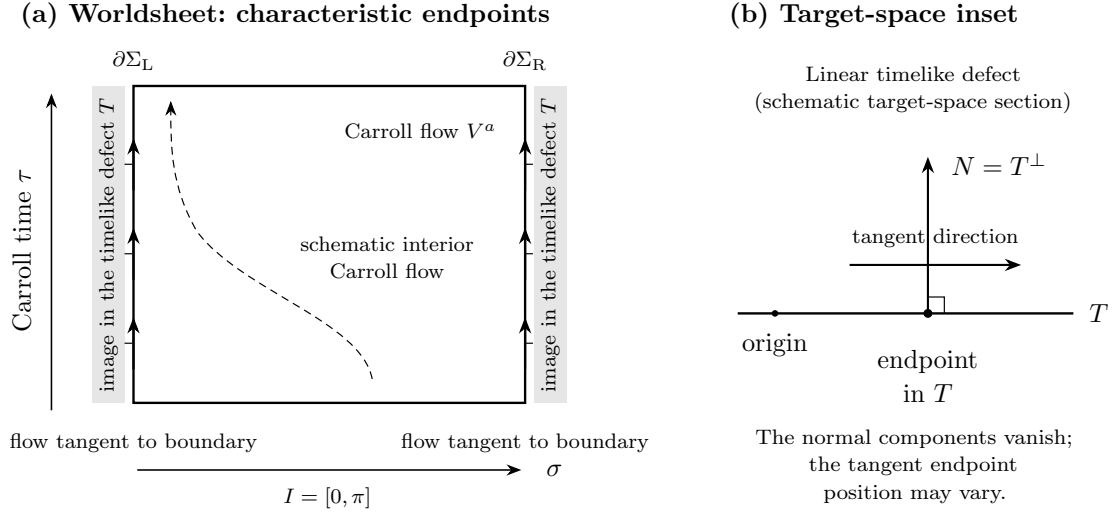
\begin{figure}[tbp]
\centering
\resizebox{\linewidth}{!}{%
  \begin{tikzpicture}[x=1cm,y=1cm,font=\small,>=Stealth,
  flow/.style={-{Stealth[length=2.2mm]},line width=.8pt},
  guide/.style={densely dashed,line width=.45pt},
  boundary/.style={line width=.9pt}]
  \path[use as bounding box] (-.35,-.65) rectangle (14.85,6.25);
  \node[anchor=west,font=\bfseries\small] at (0,6.0)
    {(a) Worldsheet: characteristic endpoints};
  \node[anchor=west,font=\bfseries\small] at (9.5,6.0)
    {(b) Target-space inset};

  \fill[black!10] (1.1,.8) rectangle (1.53,5.05);
  \fill[black!10] (7.02,.8) rectangle (7.45,5.05);
  \node[rotate=90,font=\scriptsize] at (1.315,2.92)
    {image in the timelike defect \(T\)};
  \node[rotate=90,font=\scriptsize] at (7.235,2.92)
    {image in the timelike defect \(T\)};
  \draw[boundary] (1.65,.8) rectangle (6.9,5.05);
  \node[font=\scriptsize,anchor=south] at (1.65,5.13)
    {\(\partial\Sigma_{\mathrm L}\)};
  \node[font=\scriptsize,anchor=south] at (6.9,5.13)
    {\(\partial\Sigma_{\mathrm R}\)};
  \foreach \y in {1.25,2.45,3.65}{
    \draw[flow] (1.65,\y)--(1.65,\y+.7);
    \draw[flow] (6.9,\y)--(6.9,\y+.7);
    \draw[guide] (1.53,\y+.35)--(1.65,\y+.35);
    \draw[guide] (6.9,\y+.35)--(7.02,\y+.35);
  }

  \node[anchor=east,font=\scriptsize] at (6.65,4.45)
    {Carroll flow \(V^a\)};
  \node[font=\scriptsize,anchor=north] at (1.65,.53)
    {flow tangent to boundary};
  \node[font=\scriptsize,anchor=north] at (6.9,.53)
    {flow tangent to boundary};

  \draw[guide,-{Stealth[length=2mm]}]
    (4.85,1.12) .. controls (4.7,1.9) and (3.1,2.3) ..
    (2.5,3.1) .. controls (2.15,3.7) and (2.15,4.25) .. (2.15,4.85);
  \node[align=center,font=\scriptsize,fill=white,inner sep=2pt]
    at (5.05,2.75) {schematic interior\\Carroll flow};

  \draw[-{Stealth[length=2mm]},line width=.6pt] (.55,.7)--(.55,4.95);
  \node[rotate=90] at (.15,2.85) {Carroll time \(\tau\)};
  \draw[-{Stealth[length=2mm]},line width=.6pt] (1.7,-.1)--(6.85,-.1);
  \node[anchor=west] at (7.05,-.1) {\(\sigma\)};
  \node[font=\scriptsize] at (4.25,-.47) {\(I=[0,\pi]\)};

  \draw[boundary] (9.75,2.0)--(14.25,2.0);
  \node[anchor=west] at (14.33,2.0) {\(T\)};
  \fill (10.25,2.0) circle (1.25pt);
  \node[anchor=north] at (10.25,1.86) {origin};
  \fill (12.3,2.0) circle (1.7pt);
  \draw[flow] (12.3,2.0)--(12.3,4.1);
  \node[anchor=west] at (12.48,4.0) {\(N=T^\perp\)};
  \draw[line width=.45pt] (12.3,2.22)--(12.52,2.22)--(12.52,2.0);
  \draw[flow] (11.25,2.65)--(13.55,2.65);
  \node[font=\scriptsize,anchor=south] at (12.4,2.75)
    {tangent direction};
  \node[anchor=north,align=center] at (12.3,1.65)
    {endpoint\\in \(T\)};
  \node[align=center,font=\scriptsize,text width=4.7cm] at (12.1,-.05)
    {The normal components vanish;\\the tangent endpoint position may vary.};
  \node[font=\scriptsize,align=center,text width=5cm] at (12.1,5.0)
    {Linear timelike defect\\(schematic target-space section)};
\end{tikzpicture}%
}
\caption{Characteristic endpoints and the target defect.  (a) The boundary
images lie in a timelike defect \(T\).  On the characteristic branch, the
Carroll flow \(V^a\) is tangent to each worldsheet boundary,
\(n_aV^a=0\).  The dashed interior curve illustrates a flow approaching a
direction parallel to the boundary; the boundary arrows represent the
boundary flow itself.  (b) A schematic target-space section shows an endpoint
in the linear defect \(T\) and its normal space \(N=T^\perp\).  The normal
endpoint components vanish, while the tangential position may vary.}
\label{fig:endpoint-geometry}
\end{figure}

The action fixes the boundary flux, but it does not by itself choose a target
reflection.  We therefore first state the natural variation condition for a
defect and distinguish its worldsheet branches.  Let
\begin{equation}
  Y^\mu=V^aD_aX^\mu .
  \label{eq:Y-flow-boundary}
\end{equation}
For a tangent defect variation \(\delta X^\mu=e_i{}^\mu\delta\xi^i\), the
retained endpoint term in Eq.~\eqref{eq:boundary-variation} gives
\begin{equation}
  (n_aV^a)\,e_i{}^\mu Y_\mu\big|_{\partial\Sigma}=0 .
  \label{eq:natural-tangent-flux}
\end{equation}
In adapted nondegenerate tangent coordinates,
\begin{equation}
  (n\cdot V)Y^A\big|_{\partial\Sigma}=0 .
  \label{eq:natural-tangent-condition}
\end{equation}
An endpoint is called characteristic when \(n_aV^a=0\), so the Carroll flow
is tangent to the worldsheet boundary.  A non-characteristic endpoint has
\(n_aV^a\neq0\), so the flow crosses the boundary.
For ordinary Neumann directions, the tangential component of this condition
is the familiar open-string boundary condition
\cite{CallanLNY1987,Polchinski1995}.

At a non-characteristic endpoint, \(n\cdot V\neq0\), the natural condition
requires \(Y^A=0\).  The complete spatial gluing and higher endpoint
conditions remain separate data.  At a characteristic endpoint,
\(n\cdot V=0\), the normal boundary flux vanishes automatically.  A target
gluing and a smooth reflected extension are then independent inputs.  Adding
\(Y^A=0\) on the characteristic branch selects a smaller data class.  It is
not a consequence of the natural variation on that branch.

Fixed Dirichlet components obey \(\delta X^I=0\).  A fixed position at the
target origin preserves arbitrary endpoint CW parameters.  A nonzero fixed
position restricts the parameter or requires a compensated boundary variable.
That construction is separated from the present paper.

\refstepcounter{proposition}
\emph{Proposition \theproposition\ (Nonzero endpoint flow requires a
characteristic boundary).}\label{prop:noncharacteristic-zero-flow}
Consider a smooth solution near an open segment of a worldsheet boundary
mapped to a timelike target-space defect.  Assume natural tangential endpoint
variations and no additional boundary interactions.  If the boundary is
non-characteristic, then \(Y^\mu=V^aD_aX^\mu\) vanishes on that segment and
in the adjacent one-sided neighborhood reached by its flow.  In a chart with
\(V^0\neq0\), the momentum vanishes there as well.  Thus nonzero endpoint
flow requires a characteristic boundary.

\emph{Proof.}
At a non-characteristic endpoint the natural condition gives \(Y_T=0\), so
\(Y\) belongs to the normal space of the defect.  This space is positive
definite for a timelike defect.  The bulk constraint \(Y^2=0\) therefore
implies \(Y=0\).  The embedding equation is
\begin{equation}
  (\partial_a-W_a)(V^aY_\mu)=0 .
  \label{eq:Y-transport}
\end{equation}
Along the Carroll flow it is a homogeneous transport equation.  Since the
flow is transverse to a non-characteristic boundary, the zero boundary value
extends locally along the adjacent flow lines.  Finally,
\(P_\mu=\kappa V^0Y_\mu\) gives the momentum statement.

For the linear defects through the origin used below, the normal Dirichlet
data are
\begin{equation}
  X^I\big|_{\partial\Sigma}=0.
  \label{eq:origin}
\end{equation}
These data are preserved by local CW scaling, since
\(\delta_\chi X^I=\chi X^I\).  A nonzero fixed position requires the
compensated boundary construction analyzed in Supplementary Note 1 of
Ref.~\cite{ReflectedConstraintRepro}.

For a Neumann direction, the boundary derivative has weight zero:
\begin{equation}
  \dchi(V^aD_aX^A)
  =
  -\chi V^aD_aX^A+\chi V^aD_aX^A
  =
  0.
  \label{eq:neumann-covariance}
\end{equation}
Hence its zero set is CW invariant.  For a characteristic endpoint, the relevant
object is
\begin{equation}
  \dchi(n_aV^a)=-\chi n_aV^a .
  \label{eq:null-weight}
\end{equation}
The value zero is therefore also invariant under the full CW flow.

The natural endpoint analysis has selected the flux conditions.  We next
construct the Hamiltonian system on the characteristic interval using the
natural boundary conditions and their time preservation.  A smooth reflected
subdomain is selected and analyzed separately in Sec.~\ref{sec:algebra}.

\section{Hamiltonian consistency and Carroll--Weyl symmetry}
\label{sec:canonical}

We check three requirements: differentiability of the generators, tangency
of their Hamiltonian vector fields to the boundary domain, and preservation
of that domain under time evolution.  We use the standard Dirac--Bergmann
algorithm for constrained Hamiltonian and gauge systems
\cite{Dirac1964,HenneauxTeitelboim1992}.

\subsection{The boundary domain and its time preservation}

From this point onward, the endpoint defect is a fixed nondegenerate linear
subspace \(T\) through the target origin.  Let \(N=T^\perp\), and use \(A\)
for tangent indices and \(I\) for normal indices.  We do not assume a
reflected extension.  The worldsheet is characteristic at both endpoints,
\begin{equation}
  V^1(0)=V^1(\pi)=0,
  \label{eq:characteristic-endpoints}
\end{equation}
and we work in the local chart \(V^0\neq0\).

The embedding momentum is
\begin{equation}
  P_\mu=\kappa V^0V^aD_aX_\mu .
  \label{eq:momentum}
\end{equation}
The four auxiliary momenta are primary constraints:
\begin{equation}
  \Pi_a=\frac{\partial\mathcal L}{\partial\dot V^a}=0,
  \qquad
  Q^a=\frac{\partial\mathcal L}{\partial\dot W_a}=0 .
  \label{eq:primary-momenta}
\end{equation}
The canonical two-form includes all auxiliary pairs,
\begin{equation}
  \Omega_I
  =
  \int_I\dd\sigma\,
  \left(
  \delta P_\mu\wedge\delta X^\mu
  +\delta\Pi_a\wedge\delta V^a
  +\delta Q^a\wedge\delta W_a
  \right).
  \label{eq:extended-interval-symplectic}
\end{equation}
On smooth fields this form is weakly nondegenerate: an interior
compact-support argument sets every component to zero, and smoothness extends
the conclusion to the endpoints.

The fixed defect has \(X^I|_{\partial I}=0\).  The embedding flow derived below
has
\begin{equation}
  \dot X^I
  =
  2\alpha P^I+\beta X'{}^I+\lambda X^I .
  \label{eq:normal-endpoint-flow}
\end{equation}
At an endpoint, \(V^1=0\), so \(\beta=0\), and \(X^I=0\).  Since
\(\alpha\neq0\), preservation of \(X^I|_{\partial I}=0\) gives
\begin{equation}
  P_I|_{\partial I}=0 .
  \label{eq:normal-endpoint-momentum}
\end{equation}
Boundary conditions and their time-preservation constraints are treated in
the Dirac-constraint framework of Ref.~\cite{SheikhJabbariShirzad2001}.
The derivative of \(P_I\) does not impose \(P_I'=0\):
\begin{equation}
  \dot P_I
  =
  \beta P_I'+(\beta'-\lambda)P_I .
  \label{eq:normal-momentum-preservation}
\end{equation}
Thus the characteristic interval domain is
\begin{equation}
  \widehat\Gamma_\partial
  =
  \left\{
  (X,P,V,\Pi,W,Q)\in C^\infty(I):
  X^I|_{\partial I}=P_I|_{\partial I}=V^1|_{\partial I}=0,
  \quad V^0\neq0
  \right\}.
  \label{eq:natural-interval-domain}
\end{equation}
Equation~\eqref{eq:natural-interval-domain} specifies a boundary domain.  The
primary and secondary constraints, and any quotient by gauge flows, are
separate operations.  No endpoint condition is assigned to
\(W_1,\Pi_1\), or \(Q^1\) before a reflected extension is chosen.
The multiplier \(u^1\) must vanish at both endpoints because
\begin{equation}
  \dot V^1=u^1 .
  \label{eq:V1-multiplier-boundary}
\end{equation}

\subsection{Canonical constraints and differentiable generators}

Define the connection-independent constraint densities
\begin{equation}
  \mathcal H=P^2,
  \qquad
  \mathcal D=P\cdot X',
  \qquad
  \mathcal S=P\cdot X .
  \label{eq:interval-constraint-basis}
\end{equation}
They are related to the constraint basis used in the closed CW literature by
the invertible change
\begin{equation}
  C_1=\mathcal H,
  \qquad
  C_2=\mathcal D+W_1\mathcal S,
  \qquad
  C_3=\mathcal S .
  \label{eq:C-D-basis-change}
\end{equation}
Write \(F[f]=\int_I\dd\sigma\,fF\) and set
\begin{equation}
  \alpha=\frac{1}{2\kappa(V^0)^2},
  \qquad
  \beta=-\frac{V^1}{V^0},
  \qquad
  \lambda=\beta W_1-W_0 .
  \label{eq:natural-hamiltonian-coefficients}
\end{equation}
The Legendre-reduced Hamiltonian is
\begin{equation}
  H_c
  =
  \mathcal H[\alpha]+\mathcal D[\beta]+\mathcal S[\lambda].
  \label{eq:hamiltonian}
\end{equation}
The total Hamiltonian is
\(H_T=H_c+\Pi_a[u^a]+Q^a[v_a]\).

Primary stability gives
\begin{equation}
  \begin{aligned}
  \{Q^0,H_c\}&=\mathcal S,
  &
  \{Q^1,H_c\}&=-\beta\mathcal S,\\
  \{\Pi_1,H_c\}&=\frac{\mathcal D+W_1\mathcal S}{V^0},
  &
  \{\Pi_0,H_c\}
  &=
  \frac{\mathcal H}{\kappa(V^0)^3}
  -
  \frac{V^1}{(V^0)^2}
  (\mathcal D+W_1\mathcal S).
  \end{aligned}
  \label{eq:exact-C-stability}
\end{equation}
The independent secondary constraints are therefore
\begin{equation}
  \mathcal H\approx0,
  \qquad
  \mathcal D\approx0,
  \qquad
  \mathcal S\approx0 .
  \label{eq:constraints}
\end{equation}

For smooth \(f,h\), and smooth \(g\) vanishing at both endpoints, direct
functional differentiation gives
\begin{equation}
  \begin{aligned}
  \{\mathcal H[f],\mathcal H[k]\}&=0,
  &
  \{\mathcal H[f],\mathcal D[g]\}
  &=\mathcal H[fg'-f'g],\\
  \{\mathcal H[f],\mathcal S[h]\}
  &=-2\mathcal H[fh],
  &
  \{\mathcal D[g],\mathcal D[k]\}
  &=\mathcal D[gk'-g'k],\\
  \{\mathcal D[g],\mathcal S[h]\}
  &=\mathcal S[gh'],
  &
  \{\mathcal S[h],\mathcal S[\ell]\}&=0 .
  \end{aligned}
  \label{eq:natural-smeared-algebra}
\end{equation}
The endpoint condition on \(g\) is used in
\begin{equation}
  \{\mathcal H[f],\mathcal D[g]\}
  =
  [fgP^2]_{\partial I}
  +
  \mathcal H[fg'-f'g].
  \label{eq:HD-boundary-term}
\end{equation}
The allowed class \(g|_{\partial I}=0\) is closed under the derived smearings.
In the \(\mathcal H,\mathcal D,\mathcal S\) basis the secondary constraints
are independent of the primary variables.  The primary momenta commute, and
the secondary generators have tangent Hamiltonian vector fields
\begin{equation}
  (2fP^\sharp,0),
  \qquad
  (gX',(gP)'),
  \qquad
  (hX,-hP)
  \label{eq:secondary-Hamiltonian-flows}
\end{equation}
for \((\delta X,\delta P)\).

\begin{table}[ht]
\centering
\small
\begin{tabular}{@{}p{0.12\textwidth}p{0.20\textwidth}
p{0.31\textwidth}p{0.25\textwidth}@{}}
\hline
density & definition & flow on \((X,P)\) & allowed smearing\\
\hline
\(\mathcal H=C_1\) & \(P^2\) &
\(\delta X=2fP^\sharp,\ \delta P=0\) & any smooth \(f\)\\
\(\mathcal D\) & \(C_2-W_1C_3=P\cdot X'\) &
\(\delta X=gX',\ \delta P=(gP)'\) &
smooth \(g\) with \(g(0)=g(\pi)=0\)\\
\(\mathcal S=C_3\) & \(P\cdot X\) &
\(\delta X=hX,\ \delta P=-hP\) & any smooth \(h\)\\
\hline
\end{tabular}
\caption{Constraint generators on the natural characteristic interval.
The basis change \(C_2=\mathcal D+W_1\mathcal S\) keeps every connection term.
The complete CW transformation also acts on the auxiliary fields.}
\label{tab:natural-constraints}
\end{table}

The exceptional mixed bracket appears only in the \(C_2\) basis:
\begin{equation}
  \{Q^1[u],C_2[g]\}=-\mathcal S[ug].
  \label{eq:primary-secondary-bracket}
\end{equation}
It closes on the secondary set.  Thus the four primary and three secondary
constraints are weakly first class on the regular part.

The Hamiltonian flow of \(H_c\) is
\begin{equation}
  \dot X^\mu
  =
  2\alpha P^\mu+\beta X'{}^\mu+\lambda X^\mu,
  \qquad
  \dot P_\mu
  =
  (\beta P_\mu)'-\lambda P_\mu .
  \label{eq:full-Hamiltonian-flows}
\end{equation}
Direct differentiation along this flow gives the constraint propagation
equations
\begin{equation}\small
  \begin{aligned}
  \dot{\mathcal H}
  &=
  \beta\mathcal H'+2(\beta'-\lambda)\mathcal H,\\
  \dot{\mathcal S}
  &=
  (\beta\mathcal S)'+2\alpha\mathcal H,\\
  \dot{\mathcal D}
  &=
  \beta\mathcal D'+2\beta'\mathcal D
  +\alpha\mathcal H'+2\alpha'\mathcal H+\lambda'\mathcal S .
  \end{aligned}
  \label{eq:H-propagation}
\end{equation}
The propagation equations close homogeneously on
\((\mathcal H,\mathcal D,\mathcal S)\), so evolution preserves the constraint
surface.  At a characteristic endpoint, \(\beta=0\): the spatial transport
does not cross the boundary, and the endpoint term in the Hamiltonian
variation vanishes.  The same evolution preserves the defect data
\(X^I=P_I=0\).

The only spatial boundary term in its variation is
\begin{equation}
  \delta H_c|_{\partial I}
  =
  [\beta P_\mu\delta X^\mu]_{\partial I},
  \label{eq:hamiltonian-boundary-term}
\end{equation}
which vanishes on the characteristic interval.  The flow preserves the
endpoint values in \eqref{eq:natural-interval-domain}.  Along
\(\dd\sigma/\dd\tau=-\beta\), the equations are linear.  Since the endpoint
characteristics remain fixed, smooth evolution preserves the boundary domain
locally in time.

A well-defined Hamiltonian generator on a bounded interval must have no
uncontrolled surface term, following the classical treatments of boundary
variations in constrained systems
\cite{ReggeTeitelboim1974,HarlowWu2020}.

The complete CW generator is
\begin{equation}
  G[\chi]
  =
  \int_I\dd\sigma\,
  \left[
  \chi\mathcal S-\chi\Pi_aV^a-(\partial_a\chi)Q^a
  \right].
  \label{eq:differentiable-cw-generator}
\end{equation}
For an arbitrary smooth parameter \(\chi\), its derivatives are ultralocal and
no spatial integration by parts is required.  It acts as
\begin{equation}
  \begin{aligned}
  \delta_\chi X^\mu&=\chi X^\mu,
  &
  \delta_\chi P_\mu&=-\chi P_\mu,\\
  \delta_\chi V^a&=-\chi V^a,
  &
  \delta_\chi W_a&=-\partial_a\chi,
  &
  \delta_\chi\Pi_a&=\chi\Pi_a.
  \end{aligned}
  \label{eq:generator-actions}
\end{equation}
The generator preserves the endpoint values in
\eqref{eq:natural-interval-domain} and requires no additional CW boundary
charge.  In particular,
\(\delta_\chi\mathcal D=\chi'\mathcal S\) and
\(\delta_\chi W_1=-\chi'\) cancel in \(C_2\).  The primary momenta are weakly
preserved, and
\begin{equation}
  \{G[\chi],G[\eta]\}=0,
  \qquad
  \{G[\chi],C_i[f]\}\approx0 .
  \label{eq:differentiable-generator-algebra}
\end{equation}

\subsection{Summary of Hamiltonian consistency}

The preceding calculations establish the following result.

\begin{proposition}[The characteristic interval Hamiltonian system]
\label{thm:interval-hamiltonian-system}
On \(\widehat\Gamma_\partial\), take smooth multipliers with
\(u^1|_{\partial I}=0\) in \(H_T\).  Then \(H_T\) is differentiable and its
local Hamiltonian flow preserves the domain.  The four primary constraints
and \(\mathcal H,\mathcal D,\mathcal S\) are weakly first class.  For
arbitrary smooth \(f,h\) and smooth \(g\) with \(g|_{\partial I}=0\), the
functionals \(\mathcal H[f],\mathcal D[g],\mathcal S[h]\) are differentiable,
have tangent Hamiltonian vector fields, and form the closed algebra
\eqref{eq:natural-smeared-algebra}.  The generator \(G[\chi]\) is
differentiable for arbitrary smooth \(\chi\), including nonzero endpoint
values, and requires no additional CW boundary charge.
\end{proposition}

\section{Constraint modes and their reflected realization}
\label{sec:algebra}

\subsection{Interval smearings and the O--P--A realization}

The interval algebra is already defined by the differentiable generators of
Sec.~\ref{sec:canonical}.  We first select trigonometric smearings and
compute their mode brackets.  We then impose smooth reflection and obtain
the circle-projection interpretation.  Finally, we test whether this
restriction is reachable by gauge transformations from arbitrary natural
interval data.

Define the interval modes directly from the connection-independent
constraints:
\begin{equation}\footnotesize
  O_n=-2\int_0^\pi\! \dd\sigma\,\sin(n\sigma)\mathcal D(\sigma),
  \qquad
  P_n=2\int_0^\pi\! \dd\sigma\,\cos(n\sigma)\mathcal H(\sigma),
  \qquad
  A_n=2\int_0^\pi\! \dd\sigma\,\cos(n\sigma)\mathcal S(\sigma).
  \label{eq:interval-generators}
\end{equation}
Here \(O_n\) are the sine modes of the spatial deformation constraint,
\(P_n\) are the cosine modes of the Hamiltonian constraint, and \(A_n\) are
the cosine modes of the local scale constraint.  The symbol \(P_n\) denotes a
constraint mode, not a Fourier coefficient \(p_{n\mu}\) of the momentum
field.
Substitution into Eq.~\eqref{eq:natural-smeared-algebra} gives
\begin{equation}
  \begin{aligned}
  [O_n,O_m]
  &=
  (n-m)O_{n+m}
  -(n+m)O_{n-m},\\
  [O_n,P_m]
  &=
  (n-m)P_{n+m}+(n+m)P_{n-m},\\
  [O_n,A_m]
  &=
  -mA_{n+m}+mA_{n-m},\\
  [P_n,A_m]
  &=
  -2(P_{n+m}+P_{n-m}),\\
  [P_n,P_m]
  &=
  [A_n,A_m]=0.
  \end{aligned}
  \label{eq:opa-algebra}
\end{equation}
The mode labels extend over the integers, with
\begin{equation}
  O_{-n}=-O_n,
  \qquad
  P_{-n}=P_n,
  \qquad
  A_{-n}=A_n,
  \qquad
  O_0=0.
  \label{eq:mode-reflection-rules}
\end{equation}
The O--P brackets reproduce the two-constraint BCCA sector
\cite{OpenNull,BCCAMath}.  The additional \(A_n\) modes generate local
rescalings,
\begin{equation}
  \delta_hX=hX,
  \qquad
  \delta_hP=-hP,
  \label{eq:A-local-scaling}
\end{equation}
and their mixed brackets describe the action of this scaling on the spatial
and null Hamiltonian constraints.  These trigonometric-polynomial smearings
form a closed subalgebra of the smooth interval generators.  Physical
target-space charges preserving the defect are constructed in
Sec.~\ref{sec:regular}.  The checks are reproduced in the accompanying
package \cite{ReflectedConstraintRepro}.

\subsection{Smooth reflected domains and circle normalization}

The fixed defect defines
\begin{equation}
  R=\Pi_T-\Pi_N,
  \qquad
  R^2=\mathbb{1},
  \qquad
  R^T\eta R=\eta .
  \label{eq:fixed-defect-reflection}
\end{equation}
A continuous reflection fixes endpoint values.  A smooth reflection also
matches every spatial derivative across the fold.  For even components the
odd endpoint derivatives vanish, while for odd components the even endpoint
derivatives vanish.  Thus a smooth reflected extension requires the parity
of every normal derivative.  For all \(j\geq0\),
\begin{equation}
  \begin{aligned}
  \Pi_N\partial_\sigma^{2j}X(b)&=0,
  &
  \Pi_T\partial_\sigma^{2j+1}X(b)&=0,\\
  \Pi_N\partial_\sigma^{2j}P^\sharp(b)&=0,
  &
  \Pi_T\partial_\sigma^{2j+1}P^\sharp(b)&=0,
  \end{aligned}
  \label{eq:reflected-endpoint-jets}
\end{equation}
at \(b=0,\pi\).  These conditions are equivalent to a smooth periodic
extension satisfying
\begin{equation}
  X(-\sigma)=RX(\sigma),
  \qquad
  P(-\sigma)=R^{-T}P(\sigma).
  \label{eq:reflected-field-data}
\end{equation}
\vspace{-0.65\baselineskip}
\begin{figure}[tbp]
\centering
\resizebox{\linewidth}{!}{%
  \begin{tikzpicture}[x=1cm,y=1cm,font=\small,>=Stealth]
  \path[use as bounding box] (-.25,-.25) rectangle (15.2,6.0);
  \node[anchor=west,font=\bfseries\small] at (0,5.7)
    {(a) Endpoint parities};
  \node[align=center,font=\bfseries\small] at (7.6,5.7)
    {(b) Fold and extend};
  \node[anchor=west,font=\bfseries\small] at (10.3,5.7)
    {(c) Parent circle};

  \draw[line width=.9pt] (.4,2.75)--(5.65,2.75);
  \fill (.4,2.75) circle (1.6pt);
  \fill (5.65,2.75) circle (1.6pt);
  \node[anchor=north,font=\scriptsize] at (.4,2.61)
    {\(b_{\mathrm L}\)};
  \node[anchor=north,font=\scriptsize] at (5.65,2.61)
    {\(b_{\mathrm R}\)};
  \node[font=\scriptsize] at (3.02,2.48)
    {spatial coordinate \(\sigma\)};

  \draw[line width=.8pt,domain=0:180,samples=101]
    plot ({.4+5.25*\x/180},{4.02+.56*cos(\x)});
  \node[font=\scriptsize,fill=white,inner sep=1pt] at (3.02,4.85)
    {even: \(X^A,\ P_A\)};
  \draw[line width=.7pt] (.15,4.58)--(.85,4.58);
  \draw[line width=.7pt] (5.25,3.46)--(5.9,3.46);
  \draw[densely dotted,line width=.45pt] (.4,2.85)--(.4,4.58);
  \draw[densely dotted,line width=.45pt] (5.65,2.85)--(5.65,3.46);

  \draw[densely dashed,line width=.8pt,domain=0:180,samples=101]
    plot ({.4+5.25*\x/180},{1.27+.55*sin(\x)});
  \draw[densely dotted,line width=.45pt] (.2,1.27)--(5.85,1.27);
  \fill (.4,1.27) circle (1.2pt);
  \fill (5.65,1.27) circle (1.2pt);
  \node[font=\scriptsize,fill=white,inner sep=1pt] at (3.02,.82)
    {odd: \(X^i,\ P_i\)};
  \node[font=\scriptsize,align=center] at (3.02,.2)
    {Parity is defined by reflection at each endpoint.};

  \draw[line width=.8pt] (6.65,2.75)--(8.95,2.75);
  \foreach \x in {6.65,8.95}{
    \draw[line width=.5pt] (\x,1.95)--(\x,4.35);
    \fill (\x,2.75) circle (1.6pt);
  }
  \draw[-{Stealth[open,length=2mm]},line width=.7pt]
    (6.65,3.15) .. controls (5.95,3.15) and (5.95,4.25) .. (6.65,4.25);
  \draw[-{Stealth[open,length=2mm]},line width=.7pt]
    (8.95,3.15) .. controls (9.65,3.15) and (9.65,4.25) .. (8.95,4.25);
  \node[font=\scriptsize,anchor=south] at (6.65,4.4)
    {fold at \(b_{\mathrm L}\)};
  \node[font=\scriptsize,anchor=south] at (8.95,4.4)
    {fold at \(b_{\mathrm R}\)};
  \node[font=\scriptsize,anchor=north east] at (6.53,2.55)
    {\(b_{\mathrm L}\)};
  \node[font=\scriptsize,anchor=north west] at (9.07,2.55)
    {\(b_{\mathrm R}\)};
  \node[align=center,font=\scriptsize,text width=2.9cm] at (7.8,1.05)
    {Reflect past each endpoint\\and extend periodically.};

  \draw[line width=.9pt] (13.95,3.05)
    arc[start angle=0,end angle=180,radius=1.55];
  \draw[densely dashed,line width=1pt]
    (10.85,3.05) arc[start angle=180,end angle=360,radius=1.55];
  \fill (10.85,3.05) circle (1.8pt);
  \fill (13.95,3.05) circle (1.8pt);
  \node[anchor=east,font=\scriptsize] at (10.7,3.05)
    {\(b_{\mathrm R}\)};
  \node[anchor=west,font=\scriptsize] at (14.1,3.05)
    {\(b_{\mathrm L}\)};
  \draw[-{Stealth[open,length=2.5mm]},line width=.7pt]
    (13.72,3.42) arc[start angle=16,end angle=72,radius=1.37];
  \node[align=center,font=\scriptsize] at (12.4,3.15)
    {smooth\\periodic extension};
  \node[font=\scriptsize,align=center] at (12.4,.95)
    {The interval endpoints remain\\two distinct fixed points.};
\end{tikzpicture}%
}
\caption{Endpoint reflection and the parent circle.  (a) Tangential
components \(X^A,P_A\) are even and normal components \(X^i,P_i\) are odd
under reflection about each endpoint.  The horizontal tangent marks indicate
vanishing first derivatives of the even components, while the odd components
have vanishing endpoint values.  Smooth reflection also requires all odd-order
endpoint derivatives of even components and all even-order endpoint
derivatives of odd components to vanish.  (b) Reflection across both
endpoints gives a smooth periodic extension.  (c) The interval endpoints
remain two distinct reflection fixed points.  With \(\sigma\) increasing
counterclockwise from the rightmost point of the circle,
\(b_{\mathrm L}\) at \(\sigma=0\) lies on the right and
\(b_{\mathrm R}\) at \(\sigma=\pi\) lies on the left.}
\label{fig:reflection-modes}
\end{figure}
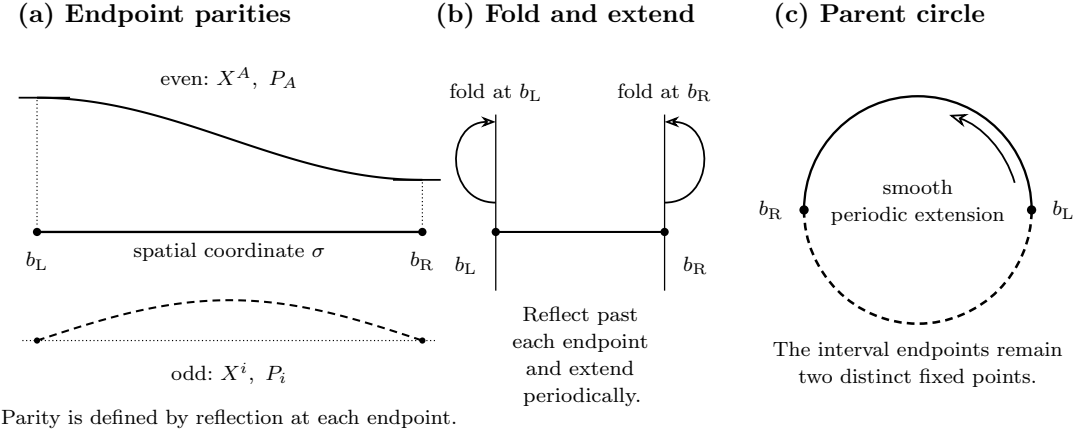

The worldsheet connection and auxiliary data must also extend smoothly.  The
even fields are
\(V^0,\Pi_0,W_0,Q^0\), while the odd fields are
\(V^1,\Pi_1,W_1,Q^1\).  The multipliers \(u^0,v_0,\chi\) are even and
\(u^1,v_1\) are odd.  Consequently \(\alpha,\lambda\) are even and
\(\beta\) is odd.  Since \(V^1\) is odd,
\begin{equation}
  V^1(-\sigma)=-V^1(\sigma)
  \quad\Longrightarrow\quad
  V^1(b)=0 .
  \label{eq:fold-characteristic}
\end{equation}
The evolution equations preserve these parities term by term:
\(\alpha,\lambda\) are even and \(\beta\) is odd.  The auxiliary equations
and even CW parameters preserve the same reflected data.  Since \(V^1\) is
odd, both fixed points are characteristic endpoints.

The full second constraint is \(C_2=\mathcal D+W_1\mathcal S\).  The
reflection makes \(\mathcal D\) odd and \(\mathcal S\) even, while the
connection parity
\begin{equation}
  W_1(-\sigma)=-W_1(\sigma)
  \label{eq:C2-connection-parity}
\end{equation}
makes \(C_2\) odd.  Therefore
\begin{equation}
  C_1(-\sigma)=C_1(\sigma),
  \qquad
  C_2(-\sigma)=-C_2(\sigma),
  \qquad
  C_3(-\sigma)=C_3(\sigma).
  \label{eq:general-constraint-parities}
\end{equation}
Thus \(C_1\) and \(C_3\) require even smearings and \(C_2\) requires an odd
smearing on the reflected domain.

\begin{table}[ht]
\centering
\small
\begin{tabular}{@{}p{0.12\textwidth}p{0.24\textwidth}
p{0.10\textwidth}p{0.27\textwidth}p{0.15\textwidth}@{}}
\hline
density & relation & parity & reflected smearing & selected basis\\
\hline
\(\mathcal H=C_1\) & \(P^2\) & even &
smooth even periodic & \(\cos n\sigma\)\\
\(\mathcal D\) & \(C_2-W_1C_3\) & odd &
smooth odd periodic & \(\sin n\sigma\)\\
\(\mathcal S=C_3\) & \(P\cdot X\) & even &
smooth even periodic & \(\cos n\sigma\)\\
\(C_2\) & \(\mathcal D+W_1\mathcal S\) & odd &
smooth odd periodic & \(\sin n\sigma\)\\
\hline
\end{tabular}
\caption{Parities and selected modes on the smooth reflected domain.  The
connection \(W_1\) is odd.  General natural-interval smearings need not obey
these parity restrictions.}
\label{tab:reflected-constraints}
\end{table}

Common target choices are \(R=-\mathbb{1}\) for an origin Dirichlet sector,
\(R=\mathbb{1}\) for the all-tangent characteristic sector, and
\(R=\operatorname{diag}(1,1,-1,-1)\) for a four-dimensional D\(_1\) example.
The explicit reduction and boost family below use the D\(_2\) reflection
\(R=\operatorname{diag}(1,1,1,-1)\).  These choices share a mode table but
need not represent the same target defect.

On \(\mathcal R_R\), the reflected currents have the parent-circle
representation in the connection-independent basis.  Define
\begin{equation}\small
  L_n=-i\int_0^{2\pi}\dd\sigma\,e^{-in\sigma}\mathcal D(\sigma),
  \quad
  M_n=\int_0^{2\pi}\dd\sigma\,e^{-in\sigma}\mathcal H(\sigma),
  \quad
  S_n=\int_0^{2\pi}\dd\sigma\,e^{-in\sigma}\mathcal S(\sigma).
  \label{eq:parent-modes}
\end{equation}
The factor \(-i\) makes the parent algebra real.  The projected combinations
are
\begin{equation}
  O_n=\frac12(L_n-L_{-n}),
  \qquad
  P_n=\frac12(M_n+M_{-n}),
  \qquad
  A_n=\frac12(S_n+S_{-n}).
  \label{eq:generators}
\end{equation}
The parent brackets follow from the same
\((\mathcal H,\mathcal D,\mathcal S)\) algebra used on the interval:
\begin{equation}\small
  [L_m,L_n]=(m-n)L_{m+n},
  \;
  [L_m,M_n]=(m-n)M_{m+n},
  \;
  [L_m,S_n]=-nS_{m+n},
  \;
  [S_m,M_n]=2M_{m+n},
  \label{eq:parent-LL-LM}
\end{equation}
with \([M_m,M_n]=[S_m,S_n]=0\).  Reflection doubles the integral of the
canonical symplectic density,
\begin{equation}
  \iota^*\Omega_{S^1}=2\Omega_I .
  \label{eq:circle-symplectic-pullback}
\end{equation}
For reflection-invariant functionals whose Hamiltonian vector fields are
tangent to \(\mathcal R_R\), this gives
\begin{equation}
  \{\iota^*F,\iota^*G\}_I
  =
  2\{F,G\}_{\rm circle}\big|_{\mathcal R_R}.
  \label{eq:half-circle-bracket}
\end{equation}
For example, projection on the circle gives
\([P_n,A_m]=-(P_{n+m}+P_{n-m})\); the interval normalization doubles this
bracket.  The zero-mode relation \([P_0,A_0]=-4P_0\) provides the same check
directly.

\subsection{An endpoint obstruction to reflection}

Reflection makes the tangent momentum even, so its first spatial derivative
vanishes at each fold.  The endpoint bivector tests whether this stationarity
can be reached along a gauge orbit.  Define
\begin{equation}
  K_b^{AB}
  =
  \left(P^AP'^{B}-P^BP'^{A}\right)\big|_b .
  \label{eq:endpoint-tangent-bivector}
\end{equation}

\begin{proposition}[Gauge-orbit obstruction for a fixed defect]
\label{prop:fixed-defect-orbit-obstruction}
Fix the nondegenerate linear defect \(T\) and its reflection
\(R=\Pi_T-\Pi_N\).  On the constraint surface of
\(\widehat\Gamma_\partial\), the zero set of \(K_b\) is invariant under
smooth endpoint-preserving gauge flows generated by the admissible
constraints.  Every smooth reflected representative for this fixed defect has
\(K_b=0\).  Hence a gauge orbit with \(K_b\neq0\) cannot meet the
corresponding reflected domain.  Such orbits exist in a regular
four-dimensional D\(_2\) sector.
\end{proposition}

An admissible constraint flow has
\(\delta P=gP'+(g'-h)P\) with \(g(b)=0\).  The endpoint bivector therefore
scales as \(\delta K_b=(3g'(b)-2h(b))K_b\).  Along a finite smooth gauge path
this multiplies \(K_b\) by a nonzero factor, so its zero set is invariant.  A
reflected representative has \(P'^A(b)=0\), which gives \(K_b=0\).

In the timelike sector with \(P^0_b\neq0\), let
\(k^a=P^a/P^0\), where \(a\) runs over the spatial tangent directions of the
defect.  Then
\begin{equation}
  K_b^{0a}=(P^0_b)^2\,\partial_\sigma k^a|_b .
  \label{eq:endpoint-direction-gradient}
\end{equation}
This is the derivative of the defect tangent direction.  An admissible
reparametrization rescales it but cannot turn a nonzero value into zero.

For an explicit obstruction, take
\begin{equation}
  R=\operatorname{diag}(1,1,1,-1),
  \quad
  V=(1,0),
  \quad
  W=0,
  \label{eq:orbit-obstruction-background}
\end{equation}
\begin{equation}
  P^\mu=\rho(1,\cos\sigma,\sin\sigma,0),
  \qquad
  X^\mu
  =
  q\sin\sigma\,e_3^\mu
  +
  \frac{\tau+\tau_0}{\kappa}P^\mu ,
  \label{eq:orbit-obstruction-solution}
\end{equation}
with \(\rho,q,\kappa>0\) and \(\tau+\tau_0>0\).  It satisfies the constraints
and natural interval conditions, while
\(K_b^{12}=\rho^2\neq0\).  Component and regularity checks are given in
Appendix~\ref{app:reflection-scope}.

\subsection{Zero modes}

The canonically normalized interval algebra gives the zero modes
\begin{equation}
  O_0=0,
  \qquad
  P_0=2\pi C_{1,0},
  \qquad
  A_0=2\pi C_{3,0}.
  \label{eq:zero-modes}
\end{equation}
The interval algebra gives the zero-mode brackets
\begin{equation}
  \begin{aligned}
  [P_0,A_0]&=-4P_0,
  \qquad
  [P_0,A_m]=-4P_m,
  \qquad
  [A_0,P_m]=4P_m,\\
  [O_m,P_0]&=2mP_m,
  \qquad
  [O_m,A_0]=0.
  \end{aligned}
  \label{eq:zero-brackets}
\end{equation}

\section{Physical degrees of freedom and conserved charges}
\label{sec:regular}

A closed constraint algebra does not by itself establish nontrivial reduced
content.  We first solve the constraints at fixed endpoint angles and identify
the remaining canonical field pair.  We then vary the endpoint directions in
a boost family and compute its conserved charge and symplectic form.

\subsection{A regular local reduced chart}
\label{sec:reduction}

We construct a local canonical chart on a regular reflected layer with fixed
endpoint angles.  The three independent first-class embedding constraints
remove three canonical pairs, leaving
\begin{equation}
  N_{\rm emb}=D-3.
  \label{eq:count}
\end{equation}
The global zero-mode conditions are discussed in
Appendix~\ref{app:half-interval}.

The additional CW constraint removes one phase-space orbit associated with
local CW scaling.  It does not imply a codimension-one null target
hypersurface.  The constraints allow \(P^2=0\) with \(X^2\neq0\), as the
explicit solution below demonstrates.

We adapt the closed CW reduction \cite{CWI,CWII} to the interval chart
\(V^0\neq0\) and construct a nonempty regular representative of the
three-constraint surface.  Take four target dimensions with
\begin{equation}
  \eta=\operatorname{diag}(-1,1,1,1),
  \qquad
  R=\operatorname{diag}(1,1,1,-1).
  \label{eq:reduction-reflection}
\end{equation}
Directions \(0,1,2\) are even and direction \(3\) is odd.  For a
future-directed branch, fix constants \(\rho>0\) and \(0<a<\pi/2\), and let
\begin{equation}
  \theta_*(\sigma)=a\cos\sigma .
  \label{eq:reference-angle}
\end{equation}
We fix the reference time, the local momentum scale, and the spatial
parametrization by
\[
  X^0=0,
  \qquad
  P_0=-\rho,
  \qquad
  \theta(P)=\theta_*(\sigma)=a\cos\sigma .
\]
The endpoint angles are held at \(a\) and \(-a\).  At fixed time in the
connection gauge \(W_a=0\), the gauge variations generated by
\(\mathcal H[f]+\mathcal D[g]+\mathcal S[h]\) are
\begin{equation}
  \delta X=2fP+gX'+hX,
  \qquad
  \delta P=(gP)'-hP .
  \label{eq:reduction-gauge-variations}
\end{equation}
The variations of the three gauge conditions are \(+2\rho f\),
\(\rho(h-g')\), and \(g\theta_*'\).  On the fixed-endpoint-angle layer,
reflection makes the angular variation vanish quadratically at the endpoints,
so division by \(\theta_*'=-a\sin\sigma\) gives a smooth odd \(g\).  In a
sufficiently small \(C^2\) neighborhood of \(\theta_*\), with nondegenerate
quadratic endpoint extrema, the finite coordinate
\(\widetilde\sigma=\arccos[\theta(\sigma)/a]\) extends smoothly by
reflection and implements the same local parametrization.

On this slice write \(P_\mu=\rho(-1,\mathbf n)\) and
\(X^\mu=(0,\mathbf x)\), with \(|\mathbf n|=1\).  The scale and spatial
constraints give
\[
  \mathbf n\cdot\mathbf x=0,
  \qquad
  \mathbf n\cdot\mathbf x'=0 .
\]
Differentiating the first relation yields
\(\mathbf n'\cdot\mathbf x=0\).  Where \(\mathbf n'\neq0\), the embedding
therefore takes the form
\[
  \mathbf x=r(\mathbf n\times\mathbf n').
\]
The fields \(\varphi\) and \(r\) encode the local angular and radial
variations.
\begin{equation}\small
  \begin{alignedat}{2}
  \mathbf m&=\mathbf n\times\mathbf n',
  \qquad
  & X_0^\mu&=(0,r\mathbf m),\\
  P_\mu&=\rho(-1,\mathbf n),
  \qquad
  & \mathbf n(\sigma)
  &=
  \bigl(
  \cos\varphi\cos\theta_*,\,
  \cos\varphi\sin\theta_*,\,
  \sin\varphi
  \bigr).
  \end{alignedat}
  \label{eq:reduction-data}
\end{equation}
Here \(r\) is smooth, \(2\pi\)-periodic, and even, while \(\varphi\) is smooth,
\(2\pi\)-periodic, odd, and satisfies \(|\varphi|<\pi/2\).
Then \(P^2=0\) and \(P\cdot X_0=0\).  Since
\(\mathbf n\cdot\mathbf m=\mathbf n\cdot\mathbf m'=0\), also
\(P\cdot X_0'=0\).  Thus the three constraints hold pointwise while \(r\)
and \(\varphi\) retain local variations.

With \(R_s=\operatorname{diag}(1,1,-1)\), the stated parities give
\(\mathbf n(-\sigma)=R_s\mathbf n(\sigma)\) and
\(\mathbf m(-\sigma)=R_s\mathbf m(\sigma)\).  Hence \(X_0\) and \(P\) obey
the target gluing.

In the temporal branch, \(X=X_0+\tau P/\kappa\) preserves the constraints,
and \(V^1=0\) makes the endpoints characteristic.  At the reference time
\(\tau=0\), set \(X_0^0=0\).  Since
\begin{equation}
  |\mathbf n'|^2
  =
  (\theta_*')^2\cos^2\varphi+(\varphi')^2>0
  \label{eq:reduction-regularity}
\end{equation}
for \(0<\sigma<\pi\), \(P\) and \(P'\) are independent.  Regularity can also
be checked directly through the constraint gradients.  On the \(W_1=0\)
slice, let
\(F=\mathcal H[f]+\mathcal D[g]+\mathcal S[h]\).  If \(dF=0\), then in the
interior
\[
  \frac{\delta F}{\delta X^\mu}
  =
  -gP_\mu'+(h-g')P_\mu=0.
\]
Linear independence of \(P\) and \(P'\) gives \(g=h=0\).  The remaining
equation
\[
  \frac{\delta F}{\delta P_\mu}=2fP^\mu=0
\]
then gives \(f=0\), because \(P\neq0\).  Smoothness extends these equalities
to the endpoints, establishing functional independence on the reflected
family.

On the reduced surface the canonical one-form becomes
\begin{equation}
  \Theta_{\rm red}
  =
  -\int_0^\pi\dd\sigma\,
  \rho r\,\mathbf m\cdot\delta\mathbf n
  =
  \int_0^\pi\dd\sigma\,p_\varphi\delta\varphi,
  \label{eq:reduction-one-form}
\end{equation}
with
\begin{equation}
  p_\varphi=-\rho r\theta_*'\cos\varphi .
  \label{eq:reduction-momentum}
\end{equation}
Conversely,
\[
  r=\frac{p_\varphi}{\rho a\sin\sigma\cos\varphi}.
\]
For smooth odd \(p_\varphi\) and \(|\varphi|<\pi/2\), the quotient extends
smoothly and evenly through the endpoints.  Thus \((\varphi,p_\varphi)\)
reconstruct the embedding on this chart.
Therefore
\begin{equation}
  \Omega_{\rm red}
  =
  \int_0^\pi\dd\sigma\,\delta p_\varphi\wedge\delta\varphi .
  \label{eq:reduction-symplectic}
\end{equation}
Both \(p_\varphi\) and \(\varphi\) are odd under the reflection.  This patch
therefore realizes one canonical embedding pair on the reflected interval.
It is a local chart near the reference profile.  The component identities and
pullback are checked in the supporting reproducibility material.

\subsection{Lorentz symmetries preserving the chosen defect}
\label{sec:lorentz}

Let \(\omega\) be a constant target Lorentz generator,
\begin{equation}
  \omega^T\eta+\eta\omega=0,
  \qquad
  [\omega,R]=0 .
  \label{eq:target-defect-lorentz-generator}
\end{equation}
The canonical moment map is
\begin{equation}
  J[\omega]
  =
  \int_I\dd\sigma\,P_\mu\omega^\mu{}_\nu X^\nu .
  \label{eq:target-lorentz-moment-map}
\end{equation}
It generates
\(\delta_\omega X=\omega X\) and
\(\delta_\omega P=-\omega^TP\), with no spatial derivative and therefore no
boundary term.  The commutation condition with \(R\) preserves both
\(T\) and \(N\), so the generator preserves the natural interval domain and
its reflected subdomain.

The constraint densities are target Lorentz scalars,
\begin{equation}
  \{J[\omega],C_i[f]\}=0,
  \qquad
  \{J[\omega],G[\chi]\}=0,
  \qquad
  \{J[\omega],J[\nu]\}=J[[\omega,\nu]] .
  \label{eq:target-moment-map-algebra}
\end{equation}
The charges commute with the auxiliary primary constraints.  For
\(j_\omega=P\omega X\), the Hamiltonian equations give
\begin{equation}
  \dot j_\omega=(\beta j_\omega)'.
  \label{eq:target-moment-map-continuity}
\end{equation}
Since \(\beta\) vanishes at both endpoints, \(J[\omega]\) is conserved.
For a timelike D\(_p\) defect, the preserved target algebra is
\begin{equation}
  \mathfrak g_R
  \simeq
  \mathfrak{so}(1,p)\oplus\mathfrak{so}(D-p-1).
  \label{eq:defect-lorentz-algebra}
\end{equation}
Different target defects can therefore share the constraint mode table while
retaining different target symmetries.

\subsection{A symplectic boost family}
\label{sec:boost}

The preceding chart fixes the endpoint angles.  We next exhibit a separate
family whose boost parameter changes the endpoint momentum direction.  This
provides a direct test that the family contains physically distinct
configurations.  The family also carries a nonzero target-space charge, and
the induced symplectic form has a nondegenerate two-dimensional pullback.
Take the same four-dimensional mixed reflection and work in the representative
\(V^a=(1,0)\), \(W_a=0\).  Define
\begin{equation}
  \begin{split}
  q(\sigma)&=q_0+q_1\cos\sigma,
  \qquad
  q_0>|q_1|>0,\\
  P_\mu&=
  \rho\bigl(-\cosh\zeta,\sinh\zeta,\cos\sigma,\sin\sigma\bigr),\\
  X_0^\mu&=
  q(\sigma)\bigl(\sinh\zeta,\cosh\zeta,0,0\bigr),\\
  X^\mu&=X_0^\mu+\frac{\tau}{\kappa}P^\mu .
  \end{split}
  \label{eq:physical-sector-data}
\end{equation}
Here \(\rho>0\), \(\kappa>0\), and \(q_1\) are fixed, while \(q_0\) and the
boost rapidity \(\zeta\) vary.  Target indices on \(P\) are raised with
\(\eta^{\mu\nu}\).  This is a separate choice of representatives.

To establish the physical meaning of \(\zeta\), use the \(W_a=0\) gauge
variation at an endpoint,
\begin{equation}
  \delta P_\mu(b)=(g'(b)-h(b))P_\mu(b).
  \label{eq:endpoint-momentum-gauge}
\end{equation}
Since \(P^0(b)\neq0\), the normalized direction has
\begin{equation}
  \delta\left(\frac{P^1}{P^0}\right)_b=0.
  \label{eq:endpoint-direction-invariance}
\end{equation}
On this family the ratio is \(\tanh\zeta\).  Changing \(\zeta\) therefore
changes the solution independently of the admissible embedding-gauge orbit.
The endpoints have spatial momentum direction
\begin{equation}
  \mathbf n_0=(\tanh\zeta,\operatorname{sech}\zeta,0),
  \qquad
  \mathbf n_\pi=(\tanh\zeta,-\operatorname{sech}\zeta,0).
  \label{eq:physical-sector-endpoint-directions}
\end{equation}
Thus \(\zeta\) changes a gauge-invariant endpoint direction.

Let
\begin{equation}
  U^\mu=(\sinh\zeta,\cosh\zeta,0,0),
  \qquad
  V^\mu=(\cosh\zeta,\sinh\zeta,0,0).
  \label{eq:physical-sector-tangent-basis}
\end{equation}
Then
\begin{equation}
  P^2=0,
  \qquad
  P\cdot U=0,
  \qquad
  P\cdot V=-\rho .
  \label{eq:physical-sector-contractions}
\end{equation}
Since \(X_0=qU\) and \(X_0'=q'U\), the first two contractions imply
\(P\cdot X_0=P\cdot X_0'=0\).  The identity \(P\cdot P'=0\) then gives
\begin{equation}
  \mathcal H=P^2=0,
  \qquad
  \mathcal D=P\cdot X'=0,
  \qquad
  \mathcal S=P\cdot X=0
  \label{eq:physical-sector-constraints}
\end{equation}
at every time.  The momentum and evolution equations follow from
\begin{equation}
  P_\mu=\kappa\eta_{\mu\nu}\dot X^\nu,
  \qquad
  \dot P_\mu=0,
  \qquad
  \ddot X^\mu=0.
  \label{eq:physical-sector-equations}
\end{equation}
The endpoints are characteristic because \(V^1=0\).

The functions \(q\) and \(\cos\sigma\) are even, while \(\sin\sigma\) is odd,
so \(X\) and \(P\) extend smoothly under the target reflection.  Moreover,
\begin{equation}
  X^2=q(\sigma)^2>0,
  \qquad
  (P')^2=\rho^2>0,
  \qquad
  P'_\mu=\rho(0,0,-\sin\sigma,\cos\sigma).
  \label{eq:physical-sector-regularity}
\end{equation}
The time component of \(P\) is nonzero, while that of \(P'\) vanishes.
Hence \(P\) and \(P'\) are independent, and the same interior test-function
argument as above establishes functional independence of the three
constraints.

The boost in the \((0,1)\) plane commutes with \(R\) and preserves the
reflection domain.  Its charge is
\begin{equation}
  J^{01}
  =
  \int_0^\pi\dd\sigma\,
  \bigl(X^0P^1-X^1P^0\bigr).
  \label{eq:boost-charge}
\end{equation}
On this family,
\begin{equation}
  X^0P^1-X^1P^0=-\rho q(\sigma),
  \qquad
  J^{01}=-\pi\rho q_0\neq0.
  \label{eq:boost-charge-value}
\end{equation}
This boost commutes with \(R\), so its charge is gauge invariant and
conserved by the general argument in Sec.~\ref{sec:lorentz}.

At the reference time, the exterior derivative of
\eqref{eq:physical-sector-data} is
\begin{equation}
  dX_0^\mu
  =
  U^\mu\dd q_0+q(\sigma)V^\mu\dd\zeta .
  \label{eq:physical-sector-differential}
\end{equation}
Using Eq.~\eqref{eq:physical-sector-contractions}, the pullback of the
half-interval one-form is
\begin{equation}
  \Theta_{\mathcal S}
  =
  \int_0^\pi\dd\sigma\,P_\mu\dd X_0^\mu
  =
  -\pi\rho q_0\,\dd\zeta
  =
  J^{01}\dd\zeta .
  \label{eq:physical-sector-one-form}
\end{equation}
Consequently
\begin{equation}
  \Omega_{\mathcal S}
  =
  \dd J^{01}\wedge\dd\zeta
  =
  -\pi\rho\,\dd q_0\wedge\dd\zeta .
  \label{eq:physical-sector-symplectic}
\end{equation}
The two-form is nondegenerate on the \((q_0,\zeta)\) surface.  In the
half-interval convention used here, it induces
\begin{equation}
  \{\zeta,J^{01}\}_{\mathcal S}=1.
  \label{eq:physical-sector-bracket}
\end{equation}
The endpoint direction \(\tanh\zeta\) is invariant under admissible gauge
flows, while the conserved charge \(J^{01}=-\pi\rho q_0\) distinguishes
changes of \(q_0\).  The pullback
\(\dd J^{01}\wedge\dd\zeta\) shows that these two physical labels form a
canonical pair on the solution family.  Here the endpoint angles vary,
complementing the field reduction at fixed endpoint angles in
Sec.~\ref{sec:reduction}.  The component identities, charge, and symplectic
pullback are checked in the accompanying reproducibility material.

\section{Discussion and conclusion}
\label{sec:conclusion}

Natural endpoint variations tie the open Carroll--Weyl theory to the direction
of its massless flow.  For a timelike defect, a nonzero flow must be tangent
to the worldsheet boundary.  This condition removes the Hamiltonian endpoint
flux and supports the complete Carroll--Weyl generator without an additional
boundary charge.

Smooth reflection preserves this dynamics and gives a circle realization of
the constraint modes.  It also restricts the allowed boundary data: an
admissible gauge transformation cannot erase a nonzero endpoint spatial
gradient of the tangent momentum direction.  The explicit orbit obstruction
therefore distinguishes reflected strings from the larger natural interval
theory.

The reflected sector retains nontrivial physical variables.  The fixed-angle
reduction exhibits a canonical field pair, while the boost family has a
gauge-invariant endpoint direction, a nonzero conserved charge, and the
symplectic form \(\dd J^{01}\wedge\dd\zeta\).  The constraint algebra, the
boundary choice, and the physical observables thus provide complementary
information about the open theory.

The next steps are to determine which higher endpoint data obstruct
reflection, analyze rank-deficient zero-flow sectors, and incorporate boundary
dynamics for displaced defects; the affine endpoint construction is developed
in Supplementary Note 1 of Ref.~\cite{ReflectedConstraintRepro}.  Quantum
BRST consistency must also
account for the interval domain and its zero-mode conditions
\cite{CWBRS,QuantumCarrollString}.  Global zero modes and null-boundary charge
observables require a separate analysis \cite{SoftChargesNull2026}.

\acknowledgments
Limin Zeng would like to express his deepest and foremost thanks to his
partner Lu Lu.

\appendix
\renewcommand{\theequation}{\Alph{section}.\arabic{equation}}
\section{Half-interval mode projection}
\label{app:half-interval}

\paragraph{Notation.}
Endpoint labels are \(b\in\{0,\pi\}\), and a prime denotes
\(\partial_\sigma\).  Uppercase \(A\) labels directions tangent to the fixed
defect \(T\), while uppercase \(I\) labels directions in \(N=T^\perp\).
The symbol \(P^\sharp\) denotes the vector with components
\(P^\mu=\eta^{\mu\nu}P_\nu\).  The connection-independent basis is
\(\mathcal H=P^2\),
\(\mathcal D=P\cdot X'\), and \(\mathcal S=P\cdot X\), related to the
covariant basis by \(C_1=\mathcal H\),
\(C_2=\mathcal D+W_1\mathcal S\), and \(C_3=\mathcal S\).  Lowercase
\(p_{n\mu}\) are Fourier coefficients of the momentum field; uppercase
\(P_n\) are modes of the Hamiltonian constraint.

For an origin Dirichlet component, the odd reflection condition gives
\begin{equation}
  X^\mu(\sigma)=\sum_{n\geq1}x_n^\mu\sin(n\sigma),
  \qquad
  P_\mu(\sigma)=\sum_{n\geq1}p_{n\mu}\sin(n\sigma).
  \label{eq:appendix-sine-modes}
\end{equation}
The half-interval pairing is
\begin{equation}
  \int_0^\pi\dd\sigma\,\sin(m\sigma)\sin(n\sigma)
  =
  \frac\pi2\delta_{mn}.
  \label{eq:appendix-sine-orthogonality}
\end{equation}
Therefore the nonzero symplectic pairing joins the same positive mode
number.  It does not pair \(n\) with \(-n\).

For an even-reflection direction, the cosine expansion is
\begin{equation}
  X^\mu(\sigma)=x_0^\mu+\sum_{n\geq1}x_n^\mu\cos(n\sigma),
  \qquad
  P_\mu(\sigma)=p_{0\mu}+\sum_{n\geq1}p_{n\mu}\cos(n\sigma).
  \label{eq:appendix-cosine-modes}
\end{equation}
The zero-mode normalization differs, but the parity of \(P_\mu X^\mu\) and
the symplectic pairing are unchanged.  A mixed target-space gluing applies
the sine and cosine rules component by component.

After splitting the target indices into even-reflection directions \(A\) and
odd-reflection directions \(i\), the corresponding mode expansions are
\begin{equation}
  \begin{aligned}
  X^A&=x_0^A+\sum_{n\geq1}x_n^A\cos(n\sigma),
  \qquad
  P_A=p_{0A}+\sum_{n\geq1}p_{nA}\cos(n\sigma),\\
  X^i&=\sum_{n\geq1}x_n^i\sin(n\sigma),
  \qquad
  P_i=\sum_{n\geq1}p_{ni}\sin(n\sigma).
  \end{aligned}
  \label{eq:zero-expansions}
\end{equation}
The zero Fourier projections of the constraints couple the even-reflection
zero modes to the nonzero modes.  Here \(C_{i,0}\) denotes the average over
the reflected circle.  The \(n=0\) Fourier components are
\begin{equation}\small
  C_{1,0}
  =
  p_{0A}p_0^A
  +
  \frac12\sum_{n\geq1}p_{n\mu}p_n^\mu,
  \qquad
  C_{2,0}=0,
  \qquad
  C_{3,0}
  =
  p_{0A}x_0^A
  +
  \frac12\sum_{n\geq1}p_{n\mu}x_n^\mu.
  \label{eq:zero-constraints}
\end{equation}
Thus the classical zero-mode data must satisfy
\begin{equation}
  C_{1,0}=0,
  \qquad
  C_{3,0}=0.
  \label{eq:zero-mode-constraints}
\end{equation}
With the half-normalized generators used here, the zero modes are
\begin{equation}
  P_0=2\pi C_{1,0},
  \qquad
  A_0=2\pi C_{3,0}.
  \label{eq:zero-density-normalization}
\end{equation}
These equations impose global conditions alongside the local field reduction.
The factors in \(P_0=2\pi C_{1,0}\) and \(A_0=2\pi C_{3,0}\) follow from the
same half-interval normalization used in the mode algebra.

\section{Constraint algebra and Jacobi identities}
\label{app:closure}

\subsection{Bulk algebra on the interval}

The mode construction uses the connection-independent densities
\(\mathcal H=P^2\), \(\mathcal D=P\cdot X'\), and \(\mathcal S=P\cdot X\).
For comparison with the covariant constraint basis, write
\(C_1=\mathcal H\), \(C_2=\mathcal D+W_1\mathcal S\), and
\(C_3=\mathcal S\).  The corresponding brackets retain the connection terms.
For test functions supported in the interior of the interval, direct
functional differentiation from
\begin{equation}
  \{X^\mu(\sigma),P_\nu(\sigma')\}
  =
\delta^\mu{}_\nu\delta(\sigma-\sigma')
  \label{eq:appendix-fundamental-bracket}
\end{equation}
includes every term generated by the connection.  For smeared quantities the
basis change is
\[
  \mathcal D[f]=C_2[f]-C_3[W_1f].
\]
A coefficient comparison with Refs.~\cite{CWI,CWII} therefore uses a common
basis.  The complete distributional \(C_2\) table is
\begin{align}
  \{C_1[f],C_1[g]\}&=0,
  \nonumber\\
  \{C_1[f],C_2[g]\}
  &=
  C_1[fg'-f'g-2W_1fg],
  \nonumber\\
  \{C_1[f],C_3[g]\}
  &=
  -2C_1[fg],
  \nonumber\\
  \{C_2[f],C_2[g]\}
  &=
  C_2[fg'-f'g],
  \nonumber\\
  \{C_2[f],C_3[g]\}
  &=
  C_3[fg'],
  \nonumber\\
  \{C_3[f],C_3[g]\}&=0.
  \label{eq:bulk-bracket-table}
\end{align}
In particular,
\[
  \{\mathcal H[f],C_2[g]\}
  =
  \mathcal H[fg'-f'g]-2\mathcal H[W_1fg].
\]
The familiar connection-independent table is recovered in the
\(\mathcal D\) basis.  In the \(C_2\)--\(C_2\) entry, the two \(W_1'fg\)
terms cancel after integration by parts, and no tertiary constraint appears.
The step-by-step functional differentiation is included in the
canonical-analysis note of the reproducibility package.  For differentiable
interval generators, the \(\mathcal D\)-basis smearings vanish at both
endpoints, removing the integration-by-parts terms.

\subsection{Reflected algebra and Jacobi checks}
\label{app:reflected-algebra}

\emph{Boundary differentiability.}
The table \eqref{eq:bulk-bracket-table} was derived with interior compact
support.  The reflected modes instead use sine and cosine functions.  For
the projected diffeomorphism mode, reflection gives
\begin{equation}
  O_n
  =
  -2\int_0^\pi\dd\sigma\,\sin(n\sigma)\mathcal D(\sigma).
  \label{eq:O-sine-representation}
\end{equation}
Its variation is
\begin{equation}
  \delta O_n
  =
  -2\int_0^\pi\dd\sigma\,\sin(n\sigma)
  \left(\delta P\cdot X'+P\cdot\delta X'\right).
  \label{eq:O-boundary-variation}
\end{equation}
After integration by parts in the second term, the boundary contribution is
\begin{equation}
  -2\bigl[\sin(n\sigma)P\cdot\delta X\bigr]_0^\pi
  =
  0,
  \label{eq:O-boundary-term}
\end{equation}
because \(\sin(n0)=\sin(n\pi)=0\).  The reflected momentum and CW modes have
the representations
\begin{equation}\small
  P_n=2\int_0^\pi\dd\sigma\,\cos(n\sigma)\mathcal H(\sigma),
  \qquad
  A_n=2\int_0^\pi\dd\sigma\,\cos(n\sigma)\mathcal S(\sigma).
  \label{eq:P-A-cosine-representation}
\end{equation}
Let \(s_n=\sin(n\sigma)\) and \(c_m=\cos(m\sigma)\).  Directly,
\begin{equation}\small
  \frac{\delta O_n}{\delta X^\mu}=2(s_nP_\mu)',
  \qquad
  \frac{\delta O_n}{\delta P_\mu}=-2s_nX'{}^\mu,
  \qquad
  \frac{\delta A_m}{\delta X^\mu}=2c_mP_\mu,
  \qquad
  \frac{\delta A_m}{\delta P_\mu}=2c_mX^\mu .
  \label{eq:O-A-functional-derivatives}
\end{equation}
Therefore
\begin{equation}\small
  \begin{aligned}
  \{O_n,A_m\}
  &=
  4\int_0^\pi\dd\sigma\,
  c_m\bigl[(s_nP_\mu)'X^\mu+s_nP_\mu X'{}^\mu\bigr]\\
  &=
  4\int_0^\pi\dd\sigma\,c_m(s_n\mathcal S)'\\
  &=
  4m\int_0^\pi\dd\sigma\,\sin(m\sigma)\sin(n\sigma)\mathcal S\\
  &=
  mA_{n-m}-mA_{n+m}.
  \end{aligned}
  \label{eq:direct-O-A}
\end{equation}
The second integration by parts has the vanishing boundary term
\(4[c_ms_n\mathcal S]_0^\pi\).

The Hamiltonian mode has
\(\delta P_n/\delta X^\mu=0\) and
\(\delta P_n/\delta P_\mu=4c_nP^\mu\), so
\begin{equation}
  \{P_n,A_m\}
  =
  -8\int_0^\pi\dd\sigma\,c_nc_mP^2
  =
  -2(P_{n+m}+P_{n-m}).
  \label{eq:direct-P-A}
\end{equation}
This calculation contains no spatial integration by parts and fixes the
half-interval normalization directly.

The \(P_n\) and \(A_n\) functionals contain no spatial derivatives and
produce no endpoint terms in their variations.  Their Hamiltonian flows are
\begin{equation}\small
  \{X^\mu,O_n\}=-2\sin(n\sigma)X'{}^\mu,
  \quad
  \{P_\mu,O_n\}=-2\partial_\sigma
  \bigl[\sin(n\sigma)P_\mu\bigr],
  \quad
  \{X^\mu,P_n\}=4\cos(n\sigma)P^\mu.
  \label{eq:O-P-boundary-flows}
\end{equation}
and
\begin{equation}
  \{X^\mu,A_n\}=2\cos(n\sigma)X^\mu,
  \qquad
  \{P_\mu,A_n\}=-2\cos(n\sigma)P_\mu.
  \label{eq:A-boundary-flow}
\end{equation}
The sine and cosine factors have the parity required by the chosen gluing.
Hence the reflected modes are differentiable, their flows preserve the
boundary domain, and the projected Poisson algebra closes.

The projected algebra is closed because every commutator reduces to the
same three generator families.  The potentially nontrivial mixed terms are
the reflected O--A and P--A commutators.  For the latter, \(P_n\) and \(A_m\)
are reflection-invariant combinations whose Hamiltonian vector fields are
tangent to the reflected domain, so
Eq.~\eqref{eq:half-circle-bracket} applies.  The projected identities are
\begin{equation}
  P_{n+m}=\frac12(M_{n+m}+M_{-n-m}),
  \qquad
  P_{n-m}=\frac12(M_{n-m}+M_{m-n}).
  \label{eq:appendix-P-identities}
\end{equation}
The circle bracket of the projected generators and the resulting
half-interval restriction are
\begin{equation}\small
  [P_n,A_m]_{\rm circle}=-(P_{n+m}+P_{n-m}),
  \qquad
  [P_n,A_m]_{\rm half}=-2(P_{n+m}+P_{n-m}).
  \label{eq:appendix-pa}
\end{equation}
The projected generators have tangent Hamiltonian vector fields and inherit
the canonical Poisson bracket, so their algebra satisfies Jacobi, including
the zero modes.  The accompanying symbolic checks audit the mode
coefficients and normalization.

\subsection{Local contact central extensions}
\label{app:central-extensions}

The reflected generators obey
\begin{equation}
  O_{-n}=-O_n,\qquad
  P_{-n}=P_n,\qquad
  A_{-n}=A_n.
  \label{eq:central-reflection-parities}
\end{equation}
The parity relations do not by themselves require a diagonal contact form.
We restrict attention to parity-compatible central terms of local
diagonal-contact form, supported at equal positive mode numbers.  Any
intermediate negative label is first reduced with
Eq.~\eqref{eq:central-reflection-parities}.  For example, an occurrence of
\(O_{-1}\) means \(-O_1\) before a contact term is evaluated.  No independent
central coefficient is introduced for a negative label.  We test the three
mixed channels with
\begin{equation}
  [G_i(n),G_j(m)]_{\rm c}
  =
  C_{ij}(n)\delta_{n,m},
  \qquad
  G_i,G_j\in\{O,P,A\},
  \label{eq:central-extension-ansatz}
\end{equation}
where \(n,m\geq1\).  Central terms in an \(OO\), \(PP\), or \(AA\) bracket
vanish by antisymmetry at \(n=m\).
We classify classical candidate central terms within this parity-compatible
local diagonal-contact ansatz, with the zero-mode actions held fixed.
General nonlocal cocycles and quantum BRST anomalies lie outside this
analysis.

Suppose first that the O--P bracket contains a central coefficient
\(C_{OP}(n)\) at equal positive mode numbers.  The zero-mode identities in
Eqs.~\eqref{eq:zero-modes} and~\eqref{eq:zero-brackets} give
\begin{equation}
  [P_n,A_0]=-4P_n,
  \qquad
  [A_0,O_n]=0.
  \label{eq:central-op-zero-brackets}
\end{equation}
The CW zero mode grades the algebra: \(P_n\) has adjoint weight four under
\(A_0\), whereas \(O_n\), \(A_n\), and central constants have weight zero.  A
central constant is therefore incompatible with the O--P and P--A brackets
when this zero-mode action is fixed.
The Jacobi identity for \((O_n,P_n,A_0)\) is
\begin{align}
  0
  &=
  [O_n,[P_n,A_0]]
  +
  [P_n,[A_0,O_n]]
  \nonumber\\
  &\quad
  +
  [A_0,[O_n,P_n]]
  =
  -4[O_n,P_n]+8nP_0
  =
  -4C_{OP}(n).
  \label{eq:central-op-veto}
\end{align}
Hence
\begin{equation}
  C_{OP}(n)=0.
  \label{eq:central-op-zero}
\end{equation}
An O--P central term is allowed in the bosonic BCCA sector of the
homogeneous boundary superconformal algebra \cite{Bagchi:2025BSCCA}.  With
the CW zero-mode action \eqref{eq:central-op-zero-brackets} held fixed,
Eq.~\eqref{eq:central-op-veto} forces this term to vanish within the local
diagonal-contact class considered here.

The same argument excludes a P--A central term.  For a coefficient
\(C_{PA}(n)\) at equal positive modes, the Jacobi identity for
\((A_0,P_n,A_n)\), together with
\begin{equation}
  [A_0,P_n]=4P_n,
  \label{eq:central-pa-zero-bracket}
\end{equation}
gives
\begin{equation}
  \small
  0
  =
  [A_0,[P_n,A_n]]
  +
  [P_n,[A_n,A_0]]
  +
  [A_n,[A_0,P_n]]
  =
  -8(P_{2n}+P_0)+4[A_n,P_n]
  =
  -4C_{PA}(n).
  \label{eq:central-pa-veto}
\end{equation}
The noncentral terms cancel, and the central part is
\(-4C_{PA}(n)\).  Hence
\begin{equation}
  C_{PA}(n)=0.
  \label{eq:central-pa-zero}
\end{equation}

Only the O--A channel remains.  Write
\(b_n=C_{OA}(n)\).  The mixed Jacobi identities for two representative
families of triples, \((O_1,A_1,O_2)\) and \((A_1,O_n,O_{n+1})\), reduce to
\begin{equation}
  b_2=2b_1,
  \qquad
  b_{n+1}=(2n+1)b_1-b_n
  \quad(n\geq2).
  \label{eq:central-oa-recurrence}
\end{equation}
Therefore
\begin{equation}
  b_n=n b_1.
  \label{eq:central-oa-linear}
\end{equation}
This linear term is a coboundary.  Define
\begin{equation}
  \Phi(A_0)=b_1,
  \qquad
  \Phi(O_n)=\Phi(P_n)=\Phi(A_{n\neq0})=0.
  \label{eq:central-oa-functional}
\end{equation}
Using Eq.~\eqref{eq:opa-algebra} and the reflection parity of \(A_n\),
\begin{equation}
  \Phi([O_n,A_m])
  =
  n b_1\delta_{n,m}.
  \label{eq:central-oa-coboundary}
\end{equation}
This relation is stated on the support of Eq.~\eqref{eq:central-extension-ansatz},
namely \(n,m\geq1\) after the negative-label reduction just described.
Thus the O--A representative in Eq.~\eqref{eq:central-oa-linear} belongs to
the trivial cohomology class and can be removed by a parity-compatible
redefinition of the zero mode.

The closed-string result of Ref.~\cite{CWBRS} agrees with this restriction.
The LL and SS central terms are odd polynomials of the parent mode label and
cancel in the reflected \(O\) and \(A\) combinations.  For the LS cocycle of
\href{https://arxiv.org/html/2608.02731v1\#S5.E16}{Eq.~(5.16)} of that
reference, the reflected difference is exactly
\(b_n=\frac12c_{LS}n\) after the half projection and half-interval
normalization.  This is the coboundary in
Eq.~\eqref{eq:central-oa-coboundary}.  Consequently the closed-string
Weyl--BMS classes do not descend as ordinary central extensions of the
minimal reflected O--P--A algebra within the local diagonal-contact class
considered here.

\section{Proof of Proposition~\ref{prop:fixed-defect-orbit-obstruction}}
\label{app:reflection-scope}

\subsection{Detailed endpoint-bivector calculation}

\begin{proof}[Proof of Proposition~\ref{prop:fixed-defect-orbit-obstruction}]
An admissible embedding-constraint flow acts on momentum by
\begin{equation}
  \delta P=(gP)'-hP=gP'+(g'-h)P,
  \qquad
  g(b)=0 .
  \label{eq:constraint-flow-momentum}
\end{equation}
Differentiating and antisymmetrizing, and using \(P\wedge P=0\), gives
\begin{equation}
  \delta(P_T\wedge P_T')
  =
  g(P_T\wedge P_T')'
  +(3g'-2h)(P_T\wedge P_T').
  \label{eq:bivector-Ward}
\end{equation}
At the endpoint,
\begin{equation}
  \delta K_b=(3g'(b)-2h(b))K_b .
  \label{eq:bivector-endpoint-flow}
\end{equation}
Primary generators and \(\mathcal H[f]\) do not change \(P\).  Along a finite
smooth gauge path, \(K_b\) is multiplied by a nonzero exponential factor.
Thus its zero set is invariant.  In the \(C_2\) basis, the endpoint value of
the effective \(h\) is unchanged because \(g(b)=0\).

A reflected representative has even tangent momentum, so
\(P'^A(b)=0\) and \(K_b=0\).  Therefore a nonzero \(K_b\) obstructs a
representative in \(\mathcal R_R\).

For existence, use the orbit in
Eq.~\eqref{eq:orbit-obstruction-solution}, with
\(\rho,q,\kappa>0\) and \(\tau+\tau_0>0\).  Direct substitution gives
\(P^2=P\cdot X=P\cdot X'=0\),
\(\dot X=P^\sharp/\kappa\), and \(\dot P=0\).  The auxiliary equations hold
because \(Y=P^\sharp/\kappa\) and
\(Y^2=Y\cdot X'=X\cdot Y=0\).  At both endpoints
\(X^3=P^3=V^1=0\).  Since \(P'^2=\rho^2\), \(P\) and \(P'\) are independent,
and the same constraint-gradient argument as in the regular reduction proves
functional independence.  Nevertheless,
\begin{equation}
  K_b^{12}
  =
  \rho^2(\cos^2b+\sin^2b)
  =
  \rho^2\neq0 .
  \label{eq:orbit-obstruction-bivector}
\end{equation}
The orbit therefore has no representative in the corresponding reflected
domain.
\end{proof}

\bibliographystyle{JHEP}
\bibliography{reflected_constraint_algebras}

\end{document}